\documentclass[a4paper, UKenglish,cleveref, autoref, thm-restate]{lipics-v2021}
\hideLIPIcs  %uncomment to remove references to LIPIcs series (logo, DOI, ...), e.g. when preparing a pre-final version to be uploaded to arXiv or another public repository
\nolinenumbers
\usepackage{xskak}
\usepackage{float}
\usetikzlibrary{calc}
\NewDocumentCommand{\pair}{mm}{%
	\StrSubstitute{#1}{-}{\text{--}}[\Xsub]%
	\StrSubstitute{#2}{-}{\text{--}}[\Ysub]%
	\StrLen{#1}[\lenX]%
	\StrLen{#2}[\lenY]%
	\IfStrEq{\lenX}{1}{%
		\IfStrEq{\lenY}{1}{%
			\IfSubStr{ABCDEFGH}{#1}{%
				\IfSubStr{12345678}{#2}{%
					\(#1#2\)%
				}{%
					\((\Xsub,\Ysub)\)%
				}%
			}{%
				\((\Xsub,\Ysub)\)%
			}%
		}{%
			\((\Xsub,\Ysub)\)%
		}%
	}{%
		\((\Xsub,\Ysub)\)%
	}%
}
\title{Fog of War Chess} %TODO Please add

\author{Matthias {Gehnen}}{RWTH Aachen University, Germany \and \url{https://tcs.rwth-aachen.de/users/gehnen/} }{gehnen@cs.rwth-aachen.de}{https://orcid.org/0000-0001-9595-2992}{}%TODO mandatory, please use full name; only 1 author per \author macro; first two parameters are mandatory, other parameters can be empty. Please provide at least the name of the affiliation and the country. The full address is optional. Use additional curly braces to indicate the correct name splitting when the last name consists of multiple name parts.
\author{Julius {Stannat}}{RWTH Aachen University, Germany {} }{julius.stannat@rwth-aachen.de}{https://orcid.org/0009-0008-8809-1133}{}

\authorrunning{M. Gehnen, J. Stannat} %TODO mandatory. First: Use abbreviated first/middle names. Second (only in severe cases): Use first author plus 'et al.'

\Copyright{Matthias Gehnen, Julius Stannat } %TODO mandatory, please use full first names. LIPIcs license is "CC-BY";  http://creativecommons.org/licenses/by/3.0/

 %TODO mandatory: Please choose ACM 2012 classifications from https://dl.acm.org/ccs/ccs_flat.cfm 

 %TODO mandatory; please add comma-separated list of keywords

\category{} %optional, e.g. invited paper

\relatedversion{} %optional, e.g. full version hosted on arXiv, HAL, or other respository/website
\EventEditors{John Q. Open and Joan R. Access}
\EventNoEds{2}
\EventLongTitle{42nd Conference on Very Important Topics (CVIT 2016)}
\EventShortTitle{CVIT 2016}
\EventAcronym{CVIT}
\EventYear{2016}
\EventDate{December 24--27, 2016}
\EventLocation{Little Whinging, United Kingdom}
\EventLogo{}
\SeriesVolume{42}
\ArticleNo{23}
\begin{document}

\maketitle

\begin{abstract}
        Fog of War chess is a popular variant of classical chess, in which both players have only partial information about the position of the opponent's pieces. This study provides the first theoretical analysis of endgames in Fog of War chess. In particular, we analyze the setups king and queen versus king, king and rook versus king, and king and two rooks versus king. We show that a king and queen can always guarantee a win against a lone king. In contrast to classical chess, a king and a rook cannot guarantee a win against a lone king. However, adding one more rook guarantees a win.
\end{abstract}

\section{Introduction}
Fog of War chess, also known as Dark chess, is a variant of the popular board game chess, where each player can only see the squares to which their pieces can legally move.
The game was invented by Jens Bæk Nielsen and Torben Osted in 1989~\cite{darkness} and is based on Kriegspiel, which is the first variant of chess with limited information, introduced in 1899~\cite{zhang2025general}.

The aspect of imperfect information adds uncertainty to the game, as players must make decisions based on incomplete information about their opponent's pieces and positions.
Recently, in June 2025, the design of an AI-strategy is described, also first theoretical thoughts about the scenario of one king versus another king are mentioned~\cite{zhang2025general}.

Even in the wider social and cultural context, Fog of War chess is recognized:
Notably, on the popular chess website chess.com, tens of thousands of players actively play Fog of War chess, making it one of the most popular chess variants~\cite{chesscom_leaderboard}.
World-class players such as Hikaru Nakamura are known for playing the variant~\cite{youtube:gmhikaru}, and it has also gained a lot of traction online with millions of views~\cite{youtube:gothamchess}.

Also, other games for which imperfect information plays an important role have already been studied in the past, including minesweeper~\cite{minesweeper}, tetris~\cite{tetris}, or other video games~\cite{videoGames}.

In this article, we will study endgames in Fog of War chess from a theoretical perspective:

The main contribution therefore, is three different endgames in Fog of War chess, where we show which combination of pieces can guarantee a win against a lone king. 

First, we will introduce the rules of Fog of War chess, including some terminology.

In the three following sections, we analyze which pieces are sufficient to guarantee a win against a lone king, given the opponent's initial king position, but without knowledge of their subsequent moves. We will first show that a king and queen can always guarantee a win. After that, we consider the setup king and rook versus a king. We will show that there are starting positions where the player with the rook cannot guarantee a win. This is in contrast to classical chess, where a king and rook can always win against a lone king, given he starts and plays optimally~\cite{rook_win}.
Lastly, we consider the setup king and two rooks versus a king. We will show that the player with the two rooks can guarantee a win.

We conclude this article with some open questions.

\subsection{Rules}
The rules of Fog of War chess are similar to those of classical chess~\cite{zhang2025general}. The chess pieces move as in classical chess.
The imperfect information changes the game as follows:
\begin{itemize}
	\item The players observe only the squares their pieces can legally move to, after each move.
	\item There is no check or checkmate; a player wins by taking the opponent's king. Thus,
	\begin{itemize}
		\item if one player moves into a check with its king, then this is legal, but immediately loses the game, as there is no better move for the opponent than taking the opponent.
		\item there is no stalemate, particularly if one player can only move to a square which is attacked by the opponent, and it is his turn, he loses.
	\end{itemize}
	\item A draw does not have to be claimed. In particular, a threefold repetition and the draw after 50 moves without a pawn move or capture must be enforced by the game host.
	\item There are additional rules for pawns:
	\begin{itemize}
		\item If a pawn's forward move is blocked by a piece, then it cannot observe that piece. But of course, he then knows that there is an opponent's piece in front of him.
		\item En passant is allowed; therefore, if a pawn can do so, the opponent's pawn is visible.
	\end{itemize}
\end{itemize}

Within this study, the pawn rules are not relevant; furthermore, we ignore the draw rules.
\subsection{Notation}

We use the classic chess notation to refer to squares on the board as depicted in figure \ref{fig:notation}.

\begin{figure}[]
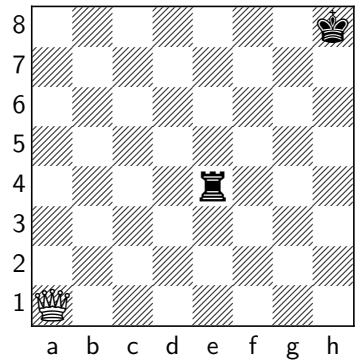

	\centering
	\newchessgame
	\setchessboard{
		showmover=false,
		boardfontsize=15pt,
		setpieces={Qa1, kh8, re4}
	}
	\chessboard
	\caption{White queen on \pair{A}{1}, Black king on \pair{H}{8}, Black rook on \pair{E}{4}}
			\label{fig:notation}
\end{figure}

The file of the chessboard is its column ($A-H$), while the rank is its row ($1-8$) \cite{chesscom_terms}. We call a square an \emph{edge-} or a \emph{corner} square, when it is at the edge (resp. corner) of the board. Without loss of generality, we will present strategies for White for the rest of this article. With, e.g., \pair{A-C}{1-3}, we denote the set of all squares on the rectangle between \pair{A}{1} and \pair{C}{3}. A \emph{configuration} is the position of all pieces on the board.
\section{Queen and King vs. King}
In classical chess, a player with a queen and a king can always force a win against a lone king \cite{rook_win}. In this section, we will show that this is also possible in Fog of War Chess:

We will present a strategy for White that guarantees a win for any given configuration and regardless of Black's moves.
Therefore, we need to show that White's strategy holds for any move Black performs.
This may lead to multiple possible positions of the Black king, of which White does not know where exactly Black is.

The winning strategy for White proceeds in two stages: 
\begin{enumerate}
	\item White moves the king onto a \emph{corner square} (\pair{A}{1}, \pair{H}{1}, \pair{A}{8}, or \pair{H}{8}) and places the queen on an adjacent square. For instance, if the king is on \pair{A}{1}, then the queen is placed on either \pair{A}{2}, \pair{B}{2}, or \pair{B}{1}. We call this a \emph{corner configuration}.
	\item White gradually reduces the set of squares available for the opponent's king until capture becomes inevitable. This is done by pushing the opponent's king towards an edge.
\end{enumerate}

\subsection{Preparing Observations}
We begin with a number of observations laying the ground for the winning strategy.

First, we note that whenever the opponent's king can be captured, White should do so.
This implies that as long as the White king or queen are not moved, they cannot be captured in the next round, as they would have been able to capture the Black king before.

Secondly, the queen can always move around their own king regardless of Blacks position.
\begin{lemma}
	 If the queen is adjacent to her own king, the queen can reach any other square adjacent to the king. If the queen gets captured, White wins immediately.
\end{lemma}
\begin{proof}
	Whenever the queen moves to a square adjacent to their own king, she cannot be captured. If Blacks king were to capture her, then he has moved his king to a place adjacent to the White king. As the next move is from White, he will capture the Black king with his own king and therefore win. It is also easy to see that White needs at most two moves to move its queen from one square adjacent to their king to another square.
\end{proof}

\subsection{Stage 1: Achieving the corner configuration}
Our goal is to reach the \emph{corner configuration} in finitely many moves. 

The objective of moving the White king to a corner and the queen to an adjacent square can be divided into different cases:
If the king is initially on one of the central squares (\pair{D}{4}, \pair{D}{5}, \pair{E}{4}, or \pair{E}{5}), then the first step is to move him out of the center. Once the king leaves the center, or if he initially does not start on a square in the center, we move the queen to an adjacent square. Finally, we coordinate the king and the queen to reach the corner configuration.

We begin by proving the last step: if the king is not in the middle of the board and the queen is located on an adjacent square, then the corner configuration can be reached.
\begin{lemma}\label{lem:kq_goal}
	The White player has a king and a queen. The Black player only has a king. It is White's turn to move. The Black king can be on any square of the board.
	
	Suppose the White king is not on one of the central squares \pair{D}{4}, \pair{D}{5}, \pair{E}{4}, or \pair{E}{5}, and the White queen is on a square adjacent to the king. Then either
	\begin{itemize}
		\item the corner configuration can be reached in finitely many moves,  or
		\item White wins.
	\end{itemize}
\end{lemma}
\begin{proof}
	Note that the maximum distance of the king to the nearest edge is at most two:
	
	\paragraph*{Case 1: Distance two from the edge.}  
	Let the king be on \pair{C}{5}, other cases are analogous. We proceed as follows:
	\begin{figure}[]
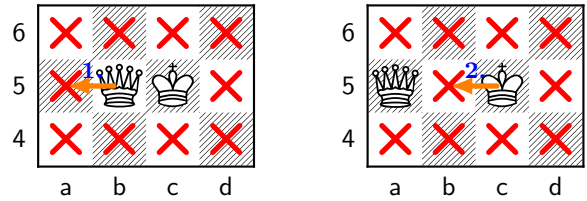

		\centering
		\newchessgame
		\setchessboard{
			printarea=a4-d6,
			showmover=false,
			setpieces={Qb5,Kc5}
		}
		\chessboard[pgfstyle=straightmove,
		linewidth=0.2ex,
		color=orange,
		markmoves={b5-a5},
		pgfstyle=text,
		color=blue,
		text=$\raisebox{1.4ex}{\small \bf 1.}$,
		markregion=b5-a5,
		pgfstyle=cross,
		color=red,
		linewidth=0.1em,
		shortenstart=0.5ex,
		shortenend=0.5ex,
		markfields={a6,b6,c6,d6,a5,a4,b4,c4,d4,d5}]
		\bigskip
		\setchessboard{
			printarea=a4-d6,
			showmover=false,
			setpieces={Qa5,Kc5},
			pgfstyle=straightmove,
			linewidth=0.2ex,
			color=orange,
			markmoves={c5-b5},
			pgfstyle=text,
			color=blue,
			text=$\raisebox{1.4ex}{\small \bf 2.}$,
			markregion=c5-b5,
			pgfstyle=cross,
			color=red,
			linewidth=0.1em,
			shortenstart=0.5ex,
			shortenend=0.5ex,
			markfields={a6,b6,c6,d6,a4,b5,b4,c4,d4,d5}
		}
		\chessboard
		\caption{The squares marked with a red cross are squares where the Black king cannot be. The orange arrow and the blue number indicate a move and its order in the sequence of moves.}
	\end{figure}
	
	\begin{enumerate}
		\setcounter{enumi}{-1}
		\item If the queen is not already on \pair{B}{5}, then move her there. After this, the opponent's king cannot be on any of the squares \pair{A}{4}, \pair{A}{5}, \pair{A}{6}, \pair{B}{4}, or \pair{B}{6}.  
		\item Next, move the queen to \pair{A}{5}. Again, she cannot be captured, as the opponent's king cannot be on any of the squares adjacent to where she moves, and now the opponent's king is also excluded from \pair{B}{4}, \pair{B}{5}, \pair{B}{6}, \pair{A}{4}, \pair{A}{6}, and \pair{C}{4}, \pair{C}{6}.  
		\item Now advance the king to \pair{B}{5}. Afterwards, the king is only one move away from the edge.  
	\end{enumerate}
	
	\paragraph*{Case 2: Distance one from the edge.}
	Once the king is within one step of the edge, he can continue without further assistance from the queen. Eventually, the king reaches the edge itself. Note that the queen cannot be on \pair{A}{5} as this would block the king from moving there, but can be placed on any other adjacent square of the king (such as \pair{A}{6}). After moving the king to the edge, the queen can be moved to an adjacent square of the king within one move. 
	
	\begin{figure}[!ht]
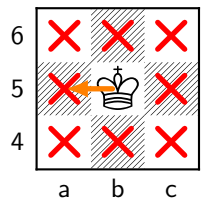

		\centering
		\newchessgame
		\setchessboard{
			printarea=a4-c6,
			showmover=false,
			setpieces={Kb5}
		}
		\chessboard[pgfstyle=straightmove,
		linewidth=0.2ex,
		color=orange,
		markmoves={b5-a5},
		pgfstyle=cross,
		color=red,
		linewidth=0.1em,
		shortenstart=0.5ex,
		shortenend=0.5ex,
		markfields={a4,a5,a6,b4,b6,c4,c5,c6}]
		\caption{The red cross indicates all squares that the Black king cannot be. All squares are covered by the White king, therefore no matter where the queen is placed, these squares are covered. The orange arrow indicates a move.}
	\end{figure}
	
	\paragraph*{Case 3: The king is at the edge}
	From there, we guide the king into a corner. Suppose the king is on \pair{A}{4}, other cases are similar. We move the queen to \pair{B}{3}, enabling the king to step downwards to \pair{A}{3}. Next, move the queen to \pair{B}{2}, after which the king descends to \pair{A}{2} and finally to \pair{A}{1}. In the end, the queen moves to \pair{A}{2}, yielding to the corner configuration.
	\begin{figure}[!ht]
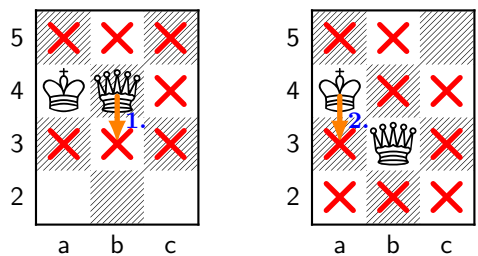

		\centering
		\newchessgame
		\setchessboard{
			printarea=a2-c5,
			showmover=false,
			setpieces={Qb4,Ka4}
		}
		\chessboard[pgfstyle=straightmove,
		linewidth=0.2ex,
		color=orange,
		markmoves={b4-b3},
		pgfstyle=text,
		color=blue,
		text=$\raisebox{-0.9ex}{\small \bf \text{ } \text{ } 1.}$,
		markregion=b4-b3,
		pgfstyle=cross,
		color=red,
		linewidth=0.1em,
		shortenstart=0.5ex,
		shortenend=0.5ex,
		markfields={a5,b5,c5,c4,a3,b3,c3}]
		\bigskip
		\setchessboard{
			printarea=a2-c5,
			showmover=false,
			setpieces={Qb3,Ka4},
			pgfstyle=straightmove,
			linewidth=0.2ex,
			color=orange,
			markmoves={a4-a3},
			pgfstyle=text,
			color=blue,
			text=$\raisebox{-0.9ex}{\small \bf \text{ } \text{ } 2.}$,
			markregion=a4-a3,
			pgfstyle=cross,
			color=red,
			linewidth=0.1em,
			shortenstart=0.5ex,
			shortenend=0.5ex,
			markfields={a5,b5,c4,a3,c3,a2,b2,c2,b4}
		}
		\chessboard
		\caption{The red cross indicates all squares where the Black king cannot be. The orange arrows and the blue numbers indicate moves and their order in the sequence of moves.}
	\end{figure}

\end{proof}
Now, we return to the initial step of moving the king out of the center, if necessary.
We show that it is always possible to move him outside the central squares and position the queen adjacent to him. This allows us to invoke Lemma~\ref{lem:kq_goal}.
\begin{lemma}[Ensuring a Safe Way]
Let the White king be on one of the central squares $\{\text{\pair{D}{4}}, \text{\pair{D}{5}}, \text{\pair{E}{4}}, \text{\pair{E}{5}}\}$, White controls a queen, and the sole remaining Black piece is a king (whose initial position is known). It is White's turn to move.
	
	Without loss of generality, White can either place its king on \pair{D}{5} while ensuring that Black's king is on \pair{D-H}{1-5} before the next move of White, or win.
	\end{lemma}
\begin{proof}
	Suppose the White king is on square \pair{D}{5} by symmetry.	

	If the Black king is not already in \pair{D-H}{1-5}, we can move our king and achieve our desired configuration up to symmetry:
	\begin{itemize}
		\item If the Black king is anywhere on the squares \pair{E-H}{6-8}, then
		\begin{itemize}
			\item If the Black king is on \pair{E}{6}, then take him.
			\item Otherwise: try to move the White king to \pair{D}{4} if not blocked by the queen. 
			\begin{itemize}
				\item If the queen is on \pair{D}{4}, then move the king to \pair{C}{4}: the king is not in the center, and the queen is adjacent to him, letting us invoke Lemma~\ref{lem:kq_goal}. 
				\item Otherwise: move the White king to \pair{D}{4}. It is Black's turn to move. If he moves, then he can only move to any square from \pair{D-H}{5-8}. Rotating the board yields the desired result.
			\end{itemize}
		\end{itemize}
		\item If the Black king is on the squares \pair{A-C}{1-4}, this is analogous to the previous case.
		\item If the Black king is on the squares \pair{A-D}{5-8}, then we can use a similar strategy. 
		\begin{itemize}
			\item If the Black king can be taken, White does so.
			\item Otherwise: try to move the White king to \pair{E}{4} if not blocked by the queen. 
			\begin{itemize}
				\item If the queen is on that square, first move to \pair{E}{5} and then to \pair{F}{4}, which can be done without a risk of capture. Again, we can continue with Lemma~\ref{lem:kq_goal}.
				\item Otherwise: move the king to \pair{E}{4}, then it is Black's turn. If he moves, then he can only safely move to any square from \pair{A-E}{4-8} except the squares adjacent to the White king. Rotate the board by 180 degrees, and we have reached the goal.\qedhere
			\end{itemize}
		\end{itemize}
	\end{itemize}

\end{proof}

Given such a configuration, we can safely move the king to any square up-left of the board, as the Black king cannot outpace the White king.

\begin{lemma}[Move away from the center]
	Let the White king be on \pair{D}{5}, White controls a queen, and the Black king is in \pair{D-H}{1-5}.
	It is White's turn to move.
	
	Then White can either reach the corner configuration or win.
\end{lemma}

\begin{proof}
	We now want to move our king, such that the queen can move to an adjacent square of the king on the next move without the risk of being captured. To ensure this, we make a case distinction based on the position of the White queen.
	
		\begin{figure}[]
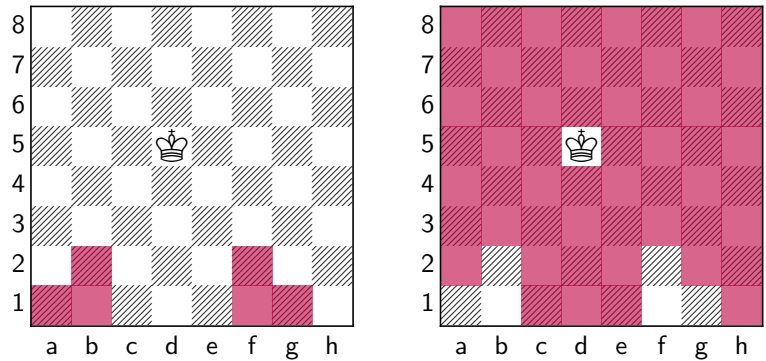

		\centering
		\newchessgame
		\setchessboard{
				showmover=false,
				boardfontsize=15pt,
	setpieces={Kd5}
}
\chessboard[pgfstyle=color,opacity=0.6,color=purple, markfields={a1,b1,b2,f1,f2,g1}]
		\setchessboard{
	showmover=false,
	setpieces={Kd5}
}
\chessboard[pgfstyle=color,opacity=0.6,color=purple, markarea=a2-a8, markarea=b3-b8, markarea=c1-c8, markarea=d1-d4, markarea=d6-d8, markarea=e1-e8, markarea=f3-f8, markarea=g2-g8, markarea=h1-h8]
		\caption{The purple squares indicate where the queen can be in cases 1 and 2 of Lemma 4}
	\end{figure}
	\paragraph*{Case 1: Queen on $\{\text{\pair{A}{1}}, \text{\pair{B}{1}}, \text{\pair{B}{2}}, \text{\pair{F}{1}}, \text{\pair{F}{2}}, \text{\pair{G}{1}}\}$}

	Then White moves the king to \pair{B}{5}. This maneuver is safe because the Black king cannot outpace the White king, as he started in the lower-right part of the board. From there, the queen can reach a square adjacent to the king within one move.
	\paragraph*{Case 2: Let the queen initially be on any other square}
	If the queen is on \pair{D}{6} or \pair{D}{7}, then move the king to \pair{C}{6}, then the queen is adjacent to the king, who is outside the center.
	
	Otherwise, we move the White king to \pair{D}{7}. This maneuver is safe because the Black king cannot outpace the White king, as he started in the lower-right part of the board. Moreover, the queen can then be repositioned onto a square adjacent to the king within one move.

	Note that if the intended path is obstructed by the Black king, then White can capture the king with the queen.

\end{proof}

It remains to consider the case where the king is not in the center of the board, and the queen is not yet adjacent to the king. We must show, that the queen can safely be brought to a square adjacent to her king. Once achieving this, Lemma~\ref{lem:kq_goal} applies.

\begin{lemma}
	Let the White player have a king not on a central square and a queen; the Black player only a king. It is White's turn. The initial position of the Black king is known.
	
	Then either the White queen can be moved to a square adjacent to the White king within at most two moves, without being captured, or White wins.
\end{lemma}
\begin{proof} 
	 If the opponent's king can be captured by the queen (or by the king), White captures and wins.
	 If there exists an adjacent square to the White king that the queen can reach in one move, White plays that move, and the condition is fulfilled in a single move.
	
	Otherwise, assume without loss of generality (by rotation) that the queen is located in the upper-right part relative to the king. Let the White king be at $K=(0,0)$ and the queen at $Q=(a,b)$ with $a,b>1$. Choose as a target the square $T=(1,1)$ (up-right from the king).
	
	\begin{figure}[]
		\centering
		\begin{tikzpicture}[scale=0.6]
		% Define board dimensions
		\def\aval{7}  % a coordinate value
		\def\bval{6}  % b coordinate value
		
		% Draw grid
		\foreach \x in {0,1,2,3,4,7,8} {
			\draw[black, thick] (\x,0) -- (\x,7);
		}
		\foreach \y in {0,1,2,3,4,6,7} {
			\draw[black, thick] (0,\y) -- (8,\y);
		}

		% Coordinate labels on x-axis
		\node[black] at (0.5, -0.3) {0};
		\node[black] at (1.5, -0.3) {1};
		\node[black] at (2.5, -0.3) {2};
		\node[black] at (3.5, -0.3) {3};
		\node[black] at (6, -0.3) {...};
		\node[black] at (7.5, -0.3) {a};
		
		% Coordinate labels on y-axis
		\node[black] at (-0.3, 0.5) {0};
		\node[black] at (-0.3, 1.5) {1};
		\node[black] at (-0.3, 2.5) {2};
		\node[black] at (-0.3, 3.5) {3};
		\node[black] at (-0.3, 4.75) {.};
		\node[black] at (-0.3, 5) {.};
		\node[black] at (-0.3, 5.25) {.};
		\node[black] at (-0.3, 6.5) {b};
		
		% Chess pieces - using chessboard package symbols
		% White king at (0,0)
		\node[font=\Huge] at (0.5, 0.5) {\WhiteKingOnWhite};
		
		% Black king at (2,2)
		\node[font=\Huge] at (2.5, 2.5) {\BlackKingOnWhite};
		
		% Queen at (a,b)
		\node[font=\Huge] at (7.5, 6.5) {\WhiteQueenOnWhite};
		
		% Red T on square (1,1)
		\node[red, font=\Large\bfseries] at (1.5, 1.5) {T};
		
		% Route 1 (purple) - from (a,b) to (1,b) to (1,1)
		\draw[purple, very thick, ->] (7.2, 6.5) -- (1.8, 6.5);
		\node[purple, font=\small] at (4.7, 6.8) {1.};
		
		\draw[purple, very thick, ->] (1.5, 6.2) -- (1.5, 1.8);
		\node[purple, font=\small] at (1.2, 4.5) {2.};
		
		% Route 2 (blue) - from (a,b) to (a,1) to (1,1)
		\draw[blue, very thick, ->] (7.5, 6.2) -- (7.5, 1.8);
		\node[blue, font=\small] at (7.8, 4.5) {1.};
		
		\draw[blue, very thick, ->] (7.2, 1.5) -- (1.8, 1.5);
		\node[blue, font=\small] at (4.5, 1.2) {2.};
		
		% Legend
		\node[black, font=\normalsize] at (9.5, 6.5) {Legend:};
		
		% Route 1 legend entry
		\fill[purple, opacity=0.8] (9, 5.9) rectangle (9.2, 6.1);
		\draw[black] (9, 5.9) rectangle (9.2, 6.1);
		\node[black, font=\small] at (9.3, 6) {:};
		\node[black, font=\small] at (10.2, 6) {Route 1};
		
		% Route 2 legend entry  
		\fill[blue, opacity=0.8] (9, 5.4) rectangle (9.2, 5.6);
		\draw[black] (9, 5.4) rectangle (9.2, 5.6);
		\node[black, font=\small] at (9.3, 5.5) {:};
		\node[black, font=\small] at (10.2, 5.5) {Route 2};
		
		\end{tikzpicture}
		\caption{The illustration shows the White king at $K=(0,0)$, the Black king at $(2,2)$, and the White queen at $Q=(a,b)$ with $a,b>1$. The target square $T=(1,1)$ is marked in red. Two possible routes for the queen to reach $T$ are shown: Route 1 in purple and Route 2 in blue.}
	\end{figure}
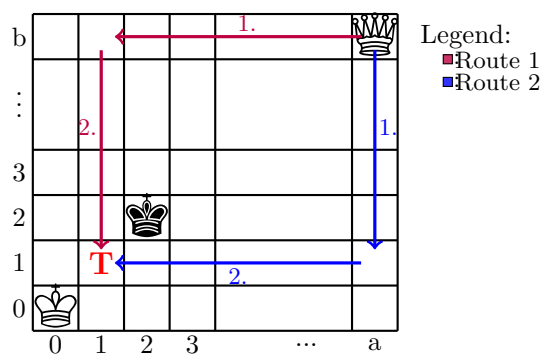
	
	From $Q=(a,b)$ to $T=(1,1)$ there are (at least) two possible routes, each consisting of a horizontal and a vertical (rook-like) move:
	\[
	(a,b)\to (1,b)\to (1,1)\quad\text{or}\quad (a,b)\to (a,1)\to (1,1).
	\]
	Since the position of the Black king is known, White can choose the route whose first waypoint is not adjacent to the Black king. If both waypoints $(1,b)$ and $(a,1)$ were adjacent to the Black king, then the Black king would necessarily be at $(2,2)$, thereby controlling the squares $(1,2),(1,3),(2,1),(3,1)$. In this case, the queen would have to be located at one of $(2,3),(3,2),(3,3)$ in order for both waypoints to be threatened. But these positions are themselves adjacent to $(2,2)$, where the queen can immediately capture the Black king.
\end{proof}

\begin{corollary}
	Starting from any position in which White controls a king and queen and moves first against a lone Black king, and White knows the Black king's initial position, it is always possible (in finitely many moves) to reach the corner configuration or to win. 
\end{corollary}

\subsection{Stage 2: Push the Black king to the edge, line by line}

We now show that it is possible to systematically reduce the cardinality of the set of squares available to the lone Black king. The reduction proceeds line by line (either rank or file), beginning with line 1 and 2. For all proofs in this part, we can assume that White starts, as black can start on an arbitrary position (within the given restrictions).
\begin{lemma}
	 White controls a king and a queen in a corner configuration, Black has a king. 
	
	Then Black can be restricted to a $6 \times 8$ square area, or White wins.
\end{lemma}
\begin{proof}
Let the White king be on \pair{A}{1} and the queen on \pair{A}{2}. Use the following strategy:
	
	\begin{enumerate}
		\item Move the queen one square to the right (first to \pair{B}{2}). She is protected by the White king, so if she is taken, White wins.
		\item Move the White king to the right. This move is safe as all adjacent positions are visible due to the queen on the position above.
		\item Repeat steps (1.), (2.) until the queen reaches \pair{H}{2}.
	\end{enumerate}
	This procedure works for the following reasons:
	\begin{itemize}
		\item If the Black king is on line 2, \pair{A-H}{2}, then the queen can take him and White wins.
		\item If the Black king is on line 1, \pair{A-H}{1}, then the queen will discover him. If the Black king moves further to the right, at some point, he is stuck and can be taken by the queen. If he takes the queen, then the White king takes and White wins.
	\end{itemize}
	
Afterwards, the Black king is not on the first or second rank, namely \pair{A-H}{1-2}.
\end{proof}

We now continue by reducing the number of possible squares available to the Black's king, line by line, until he must be placed on the last three lines. For this, let the Black king be located on an arbitrary position within \pair{A-H}{5-8}. 

 First, we move the queen and the king one step at a time to the right, until the queen sees the Black king. If the queen sees the Black king, the king can be forced upwards:

\begin{lemma}
	Let the White player have a king and a queen, and the Black player only a king. The White queen is adjacent to her own king. The Black king is located in a $6 \times 8$ square area, which is restricted by the queen. 
	
	Then White can either win the game or force the Black king upwards until restricting him to a $3 \times 8$ square field.
\end{lemma}
We show the procedure from rank 5 and above to 6 and above; the other cases are identical.
\begin{proof}
	Suppose that the White queen is located on position \pair{A}{4} and her king is on \pair{B}{3}, and the Black king is located on an arbitrary square \pair{A-H}{5-8}. Other configurations are similar.
	
	The strategy proceeds in two phases:
	\begin{itemize}
		\item First, we move the queen and the king to the right, one step at a time, until the queen sees the Black king above her.
		\item Once the queen sees the Black king, we force him upwards, ultimately restricting him to a $3 \times 8$ square area, namely \pair{A-H}{6-8}.
	\end{itemize}

	\begin{figure}[]
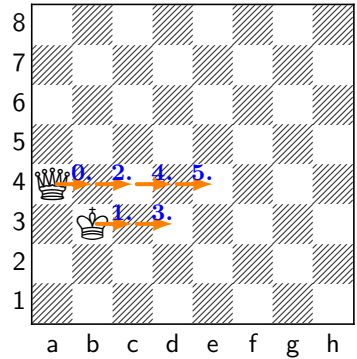

		\centering
		\newchessgame
		\setchessboard{
			showmover=false,
			boardfontsize=15pt,
			setpieces={Qa4,Kb3},
			pgfstyle=straightmove,
			linewidth=0.2ex,
			color=orange,
			markmoves={a4-b4, b3-c3, b4-c4, c3-d3, c4-d4, d4-e4},
			pgfstyle=text,
			color=blue,
			text=$\raisebox{1.4ex}{\small \bf \text{ }\text{ }0.}$,
			markregion=a4-b4,
			text=$\raisebox{1.4ex}{\small \bf \text{ }\text{ }1.}$,
			markregion=b3-c3,
			text=$\raisebox{1.4ex}{\small \bf \text{ }\text{ }2.}$,
			markregion=b4-c4,
			text=$\raisebox{1.4ex}{\small \bf \text{ }\text{ }3.}$,
			markregion=c3-d3,
			text=$\raisebox{1.4ex}{\small \bf \text{ }\text{ }4.}$,
			markregion=c4-d4,
			text=$\raisebox{1.4ex}{\small \bf \text{ }\text{ }5.}$,
			markregion=d4-e4,
			text=$\raisebox{1.4ex}{\small \bf \text{ }\text{ }6.}$}
		\chessboard
		\caption{The orange arrows and the blue numbers indicate moves and their order in the sequence of moves. It is a fixed sequence of moves, which is only interrupted if the queen sees the Black king.}
		\label{fig:chess3}
	\end{figure}
	It is White's turn to move. White follows the strategy below:
	\begin{enumerate}
		\item If White sees the Black king, then he captures immediately.  
		\item Otherwise, if the Black king is unseen, White follows a fixed sequence of moves as illustrated in Figure~\ref{fig:chess3}. 
		
		 \item If for $i=0, i=2$ or $i=4$ we see the Black king in the same rank as the queen, we move the king diagonally up-right, then follow with the queen one square up.  
		 \item If for $i=5$ we see the Black king in the same rank as the queen, we move the king diagonally up-left, then follow with the queen one square up.  
%		\end{itemize}
	\end{enumerate}
	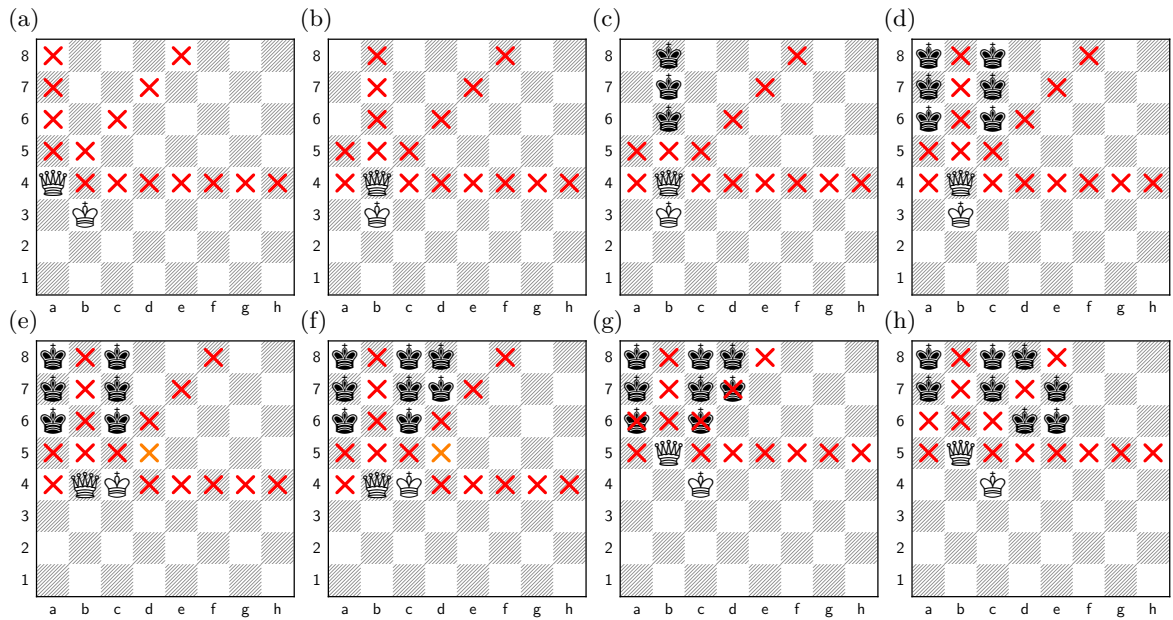
\begin{figure}[!ht]
		\centering
		\newcommand{\boardscale}{0.6}
		
		\tikzset{panel label/.style={font=\small, anchor=north west}}
		
		\begin{tikzpicture}[x=1cm,y=1cm,scale=0.8]
		
		% ========= Row 1 =========
		% (a)
		\begin{scope}[shift={(0,0)}]
		\node (A) {\scalebox{\boardscale}{%
				\chessboard[
				showmover=false,
				setpieces={Qa4,Kb3},
				pgfstyle=cross,
				color=red,
				linewidth=0.1em,
				shortenstart=0.5ex,
				shortenend=0.5ex,
				markfields={a8,a7,a6,a5,b5,c6,d7,e8,b4,c4,d4,e4,f4,g4,h4}
				]}};
		\node[panel label] at ($(A.north west)+(0.05,0)$) {(a)};
		\end{scope}
		
		% (b)
		\begin{scope}[shift={(4.8,0)}]
		\node (B) {\scalebox{\boardscale}{%
				\chessboard[
				showmover=false,
				setpieces={Qb4,Kb3},
				pgfstyle=cross,				color=red,
				linewidth=0.1em,
				shortenstart=0.5ex,
				shortenend=0.5ex,
				markfields={a4,c4,d4,e4,f4,g4,h4,a5,b5,b6,b7,b8,c5,d6,e7,f8}
				]}};
		\node[panel label] at ($(B.north west)+(0.05,0)$) {(b)};
		\end{scope}
		
		% (c)
		\begin{scope}[shift={(9.6,0)}]
		\node (C) {\scalebox{\boardscale}{%
				\chessboard[
				showmover=false,
				setpieces={Qb4,Kb3,kb6,kb7,kb8},
				pgfstyle=cross,
				color=red,
				linewidth=0.1em,
				shortenstart=0.5ex,
				shortenend=0.5ex,
				markfields={a4,c4,d4,e4,f4,g4,h4,a5,b5,c5,d6,e7,f8}
				]}};
		\node[panel label] at ($(C.north west)+(0.05,0)$) {(c)};
		\end{scope}
		
		% ========= Row 2 =========
		% (d)
		\begin{scope}[shift={(14.4,0)}]
		\node (D) {\scalebox{\boardscale}{%
				\chessboard[
				showmover=false,
				setpieces={Qb4,Kb3,ka6,ka7,ka8,kc6,kc7,kc8},
				pgfstyle=cross,
				color=red,
				linewidth=0.1em,
				shortenstart=0.5ex,
				shortenend=0.5ex,
				markfields={a4,c4,d4,e4,f4,g4,h4,a5,b5,c5,d6,e7,f8,b6,b7,b8}
				]}};
		\node[panel label] at ($(D.north west)+(0.05,0)$) {(d)};
		\end{scope}
		
		% (e) with overlay
		\begin{scope}[shift={(0,-5)}]
		\node (Ebase) {\scalebox{\boardscale}{%
				\chessboard[
				showmover=false,
				setpieces={Qb4,Kc4,ka6,ka7,ka8,kc6,kc7,kc8},
				pgfstyle=cross,
				color=red,
				linewidth=0.1em,
				shortenstart=0.5ex,
				shortenend=0.5ex,
				markfields={a4,d4,e4,f4,g4,h4,a5,b5,c5,d6,e7,f8,b6,b7,b8},
				color=orange,markfields={d5}
				]}};
		\node[panel label] at ($(Ebase.north west)+(0.05,0)$) {(e)};
		\end{scope}
		
		% (f) with overlay
		\begin{scope}[shift={(4.8,-5)}]
		\node (Fbase) {\scalebox{\boardscale}{%
				\chessboard[
				showmover=false,
				setpieces={Qb4,Kc4,ka6,ka7,ka8,kc6,kc7,kc8,kd7,kd8},
				pgfstyle=cross,
				color=red,
				linewidth=0.1em,
				shortenstart=0.5ex,
				shortenend=0.5ex,
				markfields={a4,d4,e4,f4,g4,h4,a5,b5,c5,d6,e7,f8,b6,b7,b8},
				color=orange,markfields={d5}
				]}};
		\node[panel label] at ($(Fbase.north west)+(0.05,0)$) {(f)};
		\end{scope}
		
		% ========= Row 3 =========
		% (g)
		\begin{scope}[shift={(9.6,-5)}]
		\node (G) {\scalebox{\boardscale}{%
				\chessboard[
				showmover=false,
				setpieces={Qb5,Kc4,ka6,ka7,ka8,kc6,kc7,kc8,kd7,kd8},
				pgfstyle=cross,
				color=red,
				linewidth=0.1em,
				shortenstart=0.5ex,
				shortenend=0.5ex,
				markfields={b6,b7,b8,a5,a6,c5,c6,d7,e8,d5,e5,f5,g5,h5}
				]}};
		\node[panel label] at ($(G.north west)+(0.05,0)$) {(g)};
		\end{scope}
		
		% (h)
		\begin{scope}[shift={(14.4,-5)}]
		\node (H) {\scalebox{\boardscale}{%
				\chessboard[
				showmover=false,
				setpieces={Qb5,Kc4,ka7,ka8,kc7,kc8,kd8,kd6,ke6,ke7},
				pgfstyle=cross,
				color=red,
				linewidth=0.1em,
				shortenstart=0.5ex,
				shortenend=0.5ex,
				markfields={b6,b7,b8,a5,a6,c5,c6,d7,e8,d5,e5,f5,g5,h5}
				]}};
		\node[panel label] at ($(H.north west)+(0.05,0)$) {(h)};
		\end{scope}
		
		\end{tikzpicture}
		
		\caption{The illustration shows an example sequence of moves according to the described strategy. The Black king is restricted to three ranks (6 to 8) in the end. The red crosses indicate the important squares covered by the White queen. The orange square indicates a particularly important square that is covered by the White king. There are multiple possible positions for the Black king; they are indicated by multiple Black kings. In (a), the initial position is shown, it is White's turn. In (b), the queen moved to \pair{B}{4}. In (c), the queen sees the Black king, which is on \pair{B}{6}, \pair{B}{7} or \pair{B}{8}. Black moves its king. In (e), the White king moved to \pair{C}{4}. In (g), the queen moved to \pair{B}{5}. In (h), the Black king is restricted to \pair{A-H}{6-8}.}
		\label{fig:chess-seq}
	\end{figure}
	
	\medskip
	\noindent
	\textbf{Correctness.}  
	We now prove the reduction of the Black king's possible positions:  
	\begin{itemize}
		\item In the trivial cases (Black king is directly visible, or captures the queen), White either wins by immediate capture or recaptures with the king (the queen is always protected).  

		\item At $i=2$, we consider again where the Black king can move, after the White queen saw the king. He can only move to \pair{B}{8}, \pair{B}{7}, \pair{B}{6} or \pair{D}{8}, \pair{D}{7}, \pair{D}{6}. As we move the White king, the Black king can go to \pair{A}{5}, \pair{A}{7}, \pair{A}{8}, \pair{E}{7}, \pair{E}{8} and the former squares (not \pair{E}{5}, as this is blocked by the White king). White moves the queen one square down, and it is covered by the Black king. If the Black king is on \pair{A}{5} (\pair{B}{5} respectively), then he has only two  files left. Otherwise, he has three ranks left. The cases $i=0$ and $i=4$ work identically, they also reduce the files to at most three.
		\item The case $i=5$ to $i=4$ is identical up to symmetry; but needs to be dealt mirrored as then the three files to the right side can be guaranteed.
		\item If the Black king was not seen at least at $i=5$, then he must be on the rightmost three files, and we restrict his field to $3 \times 8$ with the queen.\qedhere
	\end{itemize}
\end{proof}
In the following, we will show how to reduce the number of lines from three to two:
\begin{lemma}\label{queen3}
	Let the Black king be located on an arbitrary square within a $3 \times 8$ square area. The White player has a king and a queen, and the Black player only a king. The White queen is adjacent to her own king and restricts the Black king to the $3 \times 8$ field. 
	
	Then White can either win the game or restrict the Black king to $2 \times 8$ squares.
\end{lemma}
\begin{proof}
	We use a similar strategy by moving the White pieces to the right, until the queen sees the Black king. Then we force him upwards, restricting his position to a $2 \times 8$ area.
	
	Let the White queen be located on position \pair{A}{5} and her king be on \pair{B}{4}, and the Black king be located on an arbitrary square \pair{A-H}{6-8}. We only consider the moves of the queen moving from $A$ to $D$, as $E$ to $H$ are analogous by symmetry.
	\begin{itemize}
		\item If White sees the Black king on the $A$-file, then take.
		\item If this is not the case, White moves to \pair{B}{5}.
		\item Then White might see the opponent (queen on \pair{B}{5}).
		\begin{itemize}
			\item If White sees the opponent on the $B$-file, then it is Black's turn. Black moves to either \pair{A}{7}, \pair{A}{8}, \pair{C}{7}, or \pair{C}{8}. Then White moves the king upright to \pair{C}{5} and thereby blocks \pair{D}{6}. Black can now be on any of the squares before or on \pair{D}{8}. Move the queen up, and the Black king is restricted to the 7th and 8th rank.		
			\item Otherwise, White moves the king to the right, and then the queen one square to the right.
		\end{itemize}
		\item Then White might see the opponent (queen on \pair{C}{5}).
		\begin{itemize}
			\item If White sees the Black king on the $C$-file, then Black can move to \pair{B}{7}, \pair{B}{8}, \pair{D}{7}, or \pair{D}{8}. Then White moves the king upright, and Black can also be on to \pair{A}{6}, \pair{A}{8}, \pair{E}{8}. White then moves the queen one square up.
			\begin{itemize}
 				\item If the queen does not see the Black king on \pair{A}{6}; Black is on the 7th or 8th rank.
 				\item If the Black king is not on \pair{A}{6}; Black has only two files left, reaching our goal.
			\end{itemize}
			\item Otherwise, White moves the king and then the queen to the right.
		\end{itemize}
		\item White might now see the opponent (queen on \pair{D}{5}).
		\begin{itemize}
			\item If White sees the opponent on the $D$-rank, then move the Black king up-right, then the opponent might be on \pair{B}{6}, \pair{B}{8}, \pair{C}{7}, \pair{C}{8}, \pair{E}{7}, \pair{E}{8}, or \pair{F}{8} (as \pair{F}{6} is covered by the White king). Then move the queen up.
			\begin{itemize}
				\item If the queen sees the Black king on \pair{B}{6}, then we can apply the case above (queen on \pair{C}{5}), but with the board rotated.
				\item Otherwise, the Black king is on the 7th or 8th rank, restricted by the queen.
			\end{itemize}
			\item Otherwise, move the king and then the queen to the right.
		\end{itemize}
	\end{itemize}

	The remaining cases are analogous (differing just in moving the White king up-left).
	
	Hence, White can always force the Black king to the last and second-last rank, or win.
\end{proof}

We now want to force the Black king to have only one rank (or file) left.

\begin{lemma}\label{queen2}
	Let the Black king be located within a restricted $2 \times 8$ square area. The White player has a king and a queen, and the Black player only a king. The White queen is adjacent to her own king. The queen restricts the Black king to this area.
	
	Then White can either win the game or force the Black king to a $1 \times 8$ square area.
\end{lemma}
\begin{proof}
	As in the lemmas before, we will move the king and the queen to the right. Again, the queen will eventually see the Black king on a field above her (which must be in the topmost line in this case). Assume the White queen sees the king on square
	\begin{enumerate}
		\item \pair{A}{8}. Take the Black king, as it is White's turn.
		\item \pair{B}{8}.  Move the king up-right in the next turn. This leads to a stalemate, with the Black king to play; therefore, Black will be taken in the next move.
		\item \pair{C}{8}. Again, move the king upright. The Black king can only be on \pair{A}{7} (or taken), thus we move the queen to \pair{B}{5} and follow with the king. Thus, Black is restricted to file A.
		\item \pair{D}{8}. Black moves to either \pair{E}{8} or \pair{C}{8}. Then White moves the king upright. \textbf{We do not move the queen up!} If the opponent is on \pair{E}{8}, then he moves into a square which is covered by White, leading to a victory of White. Otherwise, if the Black king is not on \pair{E}{8}, then he is on \pair{C}{8}. We again perform our strategy. As the Black king is now on any of the squares \pair{A-C}{7-8}, we can deal with this case as in the previous cases. In all those cases, we can force the Black king to the 8th rank or win.
		\item \pair{E-G}{1}. Those cases are identical due to symmetry.\qedhere
	\end{enumerate}
	
	Hence, we can force the Black king to the last rank or win.
\end{proof}

Only one line is left; we now show a strategy that guarantees a win for White.
\begin{lemma}\label{queen1}
	
	Let the Black king be located within a $1 \times 8$ square area. The White player has a king and a queen, and the Black player only a king. It is White's turn to move. The White queen is adjacent to her own king. The queen restricts the Black king to this area. In particular, if the White queen does not move, the Black king cannot leave this area.
	
	Then there exists a deterministic strategy for White that wins the game.
	
\end{lemma}
\begin{proof}
	We start with the queen on \pair{A}{7} and the king on \pair{A}{6}. The Black king is on \pair{A}{1-8}.
	
	We move the queen gradually to the right, while the king follows her. 
	If White sees the opponent, he can escape to the right until he eventually he runs out squares and loses.
\end{proof}

\begin{corollary}
	If the player with the queen starts and plays according to the strategy outlined in this section, then he can always win against the player with only a king, given he knows the initial position of the opponent's king and starts.
\end{corollary}
Therefore, the player with the queen can always win against the player with only a king. 

\section{Rook and King vs. King}
In this section, we will show that a player with the rook does not always have a winning strategy against a player with only a king. This differs from classical chess, where a player with a rook can always force a win if he starts \cite{rook_win}!

In other words, there exists a configuration (starting position of the pieces), such that for any strategy of the player with the rook, there always exists a move for the player with only a king, such that he does not lose.
As in the section above, we assume that White has material advantage, knows the initial position of Black, and starts.

If White has a strategy to always win, he must win regardless of how the Black king moves. The Black king might be lucky and always guess the right move, such that White does not win. Therefore, we may assume that Black knows the position of White's pieces.

A first observation is that if Black captures a rook which not immediately covered, White will no longer have chances to win (as Black knows the White position).

The proof can be divided into two parts:
First, we will show that Blacks king cannot be taken if he is not at an edge. After that, we will show that he cannot be pushed to an edge.

\begin{figure}[H]
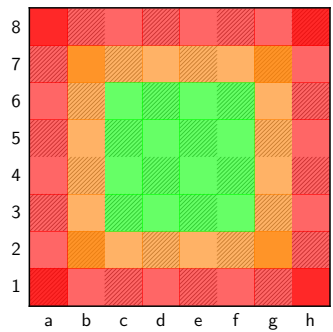

	\centering
	\scalebox{0.7}{%
		\newchessgame
		\setchessboard{setpieces={},
			showmover=false}
		\chessboard[pgfstyle=color,opacity=0.6,color=orange, markarea=b2-b7, markarea=b2-g2, markarea=g2-g7, markarea=b7-g7,
		color=green,
		markarea=c3-f6,
		color=red,
		markarea=a1-h1, markarea=a1-a8, markarea=a8-h8, markarea=h1-h8]
	}
	\caption{The green squares mark the safe space \pair{C-F}{3-6}. The orange squares indicate the danger-zone, the red squares indicate the death-zone}
\end{figure}

We claim that if the Black king is not at an edge, he cannot be taken in the next move:
\begin{lemma}
	Let it be Black's turn. And let the Black king be anywhere not on the edge of the board, i.e. \pair{B-G}{2-7}. 
	Then White cannot win in his next move.
\end{lemma}
\begin{proof}	
	Assume White wins in his next move.
	Therefore, White  must achieve either a stalemate or checkmate (as Black does not move voluntarily into a square covered by White):

	 As Black is not at the edge of the board, White has to cover all nine squares adjacent to the Black king. This is not possible. The rook can only cover five squares:
	
	\begin{figure}[!ht]
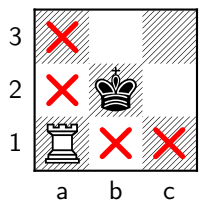

		\centering
		\newchessgame
		\setchessboard{
			printarea=a1-c3,
			showmover=false,
			setpieces={kb2, Ra1},
			pgfstyle=cross,
			color=red,
			linewidth=0.1em,
			shortenstart=0.5ex,
			shortenend=0.5ex,
			markfields={a2, a3, b1, c1}
		}
		\chessboard
		\caption{The rook can only cover five squares adjacent to the Black king. The red crosses indicate the squares covered by the rook.}
	\end{figure}
	
	Besides that, the rook can be taken by the Black king: there exists no way to cover the remaining squares (here \pair{B}{3}, \pair{C}{2}, \pair{C}{3}) with the White king without placing him on an adjacent square to the Black king (where Black can win then).
\end{proof}

It remains to show that the Black player cannot be pushed to an edge.

\begin{lemma}\label{kingpush}
	Let it be White's turn. Let Black be on a square which is not covered by White. And let Black start on any of the squares \pair{C-F}{3-6}.
	
	Then either Black cannot be forced onto an edge, or White cannot win.
\end{lemma}
\begin{proof}
	We will show that if the White player starts (and cannot immediately take the Black king), then the Black king can always remain either inside the safe space or on squares from which his moves cannot be reduced to one edge (undesired square).
	Black always tries to prevent this by going to a square inside the safe-zone whenever there exist uncovered ones.
	
	We distinguish between a corner case and an edge-but-not-corner case. The corner squares of the safe space are \pair{C}{3}, \pair{C}{6}, \pair{F}{3}, and \pair{F}{6}. The edge-but-not-corner squares are all other squares on the edge of the safe space.

	\paragraph*{Case 1: The Black king is not in a corner of the safe space:} 
	If he is somewhere on \pair{D-E}{4-5}, then White cannot force him out of the safe space in the next move, and he cannot be captured if he is not on the edge.
	
	Thus, the Black king is at the edge of the safe space, but not in a corner. To be forced down from here, the board must look like Figure~\ref{blackatedge}.
	
	\begin{figure}[!ht]
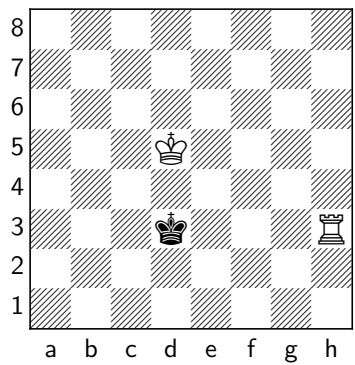

		\centering
		\newchessgame
		\setchessboard{
			%printarea=a1-e5,
			showmover=false,
			boardfontsize=15pt,
			setpieces={kd3, Rh3, Kd5}
		}
		\chessboard
		\caption{The only possibility Black can be forced down by White if he is not in a corner of the safe space (\pair{C}{3}, \pair{C}{6}, \pair{F}{3}, or \pair{F}{6}).}
		\label{blackatedge}
	\end{figure}
	
	Black can only be forced out of the safe space, if the rook attacks the rank or file of the Black king, and the White king covers the remaining squares inside the safe space. 
	
	Let the White king be on \pair{D}{5}, the rook on \pair{H}{3}, and the Black king on \pair{D}{3} (as on Figure~\ref{blackatedge}). 
	If the Black king is forced down, he can always move to \pair{D}{2}, or \pair{E}{2}, regardless. White cannot know where the Black king moved; therefore, he might be on any of those squares.
	
	\begin{itemize}
		\item Then White can move the king, afterwards Black can move from \pair{D}{2} to \pair{E}{2}, \pair{D}{1}, or \pair{E}{1} and from \pair{E}{2} to \pair{D}{2}, which allows the application of Lemma~\ref{lem:***}
		\item If White moves the rook up or down: Black returns to the safe space. If the rook is moved left or right, either the rook is captured or the Black king can move as in the case the king moved, therefore we can apply Lemma~\ref{lem:***}.
		\end{itemize}

	\paragraph*{Case 2: Black is forced outside the safe space from a corner of the safe space:}
	
	Without loss of generality, let this corner be \pair{C}{3}.	To force the Black king outside the safe space, White has to cover all squares adjacent to the Black king in the safe space:
	
	We note that the White king must always cover at least one square in the safe space adjacent to the Black king. Otherwise, the rook must be on \pair{D}{4} to cover all. Then either the king must be adjacent (but then he would cover \pair{D}{4}), or Black can take the rook.
	
	We now consider all combinations of the White king's and rook's positions that cover all safe spaces adjacent to the Black king. We then analyze the situation after White's next move.
	\begin{enumerate}
		\item Place the White king on \pair{B}{5}. Then \pair{C}{4} is covered. The rook then has to cover the $D$-file.
		\begin{enumerate}
			\item If the rook is placed on \pair{D}{2}, \pair{D}{3}, or \pair{D}{4}, the Black king would capture it.
			\item If the rook is placed on \pair{D}{1}, then Black moves to \pair{C}{2}; the rook then must retreat. After that, Black can return to \pair{C}{3} (if the rook moved to \pair{D}{3}, Black would capture it).
			\item If the rook is placed on \pair{D}{5-8}, the Black king could move to any square among \pair{B}{3}, \pair{B}{2}, or \pair{C}{2}, which introduces uncertainty for White. White is to move.
			\begin{enumerate}
				\item Let White move the king. If the king moves downward, Black can capture, as he might be on \pair{B}{3}. If the White king moves elsewhere, the Black king returns to \pair{C}{3}.
				\item Let White move the rook. If the rook moves up, to the right, or down one square, Black can return to \pair{C}{3}. If the rook is moved further down, Black, being on \pair{C}{2}, captures. If the rook is moved left, then Black, being on \pair{C}{2}, moves to \pair{D}{3}.
			\end{enumerate}
		\end{enumerate}
		\item  White king on \pair{C}{5}. Then \pair{C}{4} and \pair{D}{4} are covered. The rook only needs to cover \pair{D}{3}.
		\begin{enumerate}
			\item Rook on \pair{D}{4}: White moving anything anywhere leads to Black moving back to the safe space or capturing material of White.
			\item Rook on \pair{D}{2} or \pair{D}{3}: then the Black king captures it.
			\item Rook on \pair{D}{1}: then Black moves down to \pair{C}{2}. The White rook must move; he can only block the safe space \pair{C}{3} by going to \pair{D}{3}, where the Black king captures it. 
			\item Rook on \pair{D}{5-8}. Black moves to one of \pair{B}{2}, \pair{B}{3}, \pair{C}{2}. 
			\begin{enumerate}
				\item If White moves the rook, Black can either capture it or move back to \pair{C}{3}. 
				\item If the king moves down-right, the Black king could move to any of \pair{A-B}{1-4}, \pair{C-D}{1-2}. This case is discussed in Lemma~\ref{lem:*}.
				\item If White moves the king anywhere else, then Black can return to \pair{C}{3}.
			\end{enumerate}
			\item Rook on \pair{A}{3}: Then moving the Black king would threaten the rook, forcing it to move; after this, Black can either capture the rook or return to \pair{C}{3}.
			\item Rook on \pair{B}{3}, or \pair{D}{3}: in each case, the Black king captures it.
			\item Rook on \pair{E-H}{3}: Black moves to either \pair{B}{2}, \pair{C}{2}, or \pair{D}{2}. White is to move.
			\begin{itemize}
				\item Moving the rook up or down allows the Black king to return to \pair{C}{3}.
				\item Moving the king lets the Black king expand to \pair{A}{1}, \pair{A}{2}, \pair{B}{1}, \pair{C}{1}, or \pair{D}{1}, while also possibly being on \pair{C}{2}, \pair{D}{2}. This case is discussed in Lemma~\ref{lem:**}.
				\item  Moving the rook to the left or right either results in the rook being captured or restricting the Black king, as in the last case. This case is discussed in Lemma~\ref{lem:**}.
			\end{itemize}
		\end{enumerate}
		\item Place the White king on \pair{D}{5}. This case is analogous to (\textbf{Case 2}.).
		\item Place the White king on \pair{E}{5}. Then the rook must be on \pair{D}{4}. In this case, the Black king moves to either \pair{B}{3}, \pair{B}{2}, or \pair{C}{2}.
		\begin{figure}[!ht]
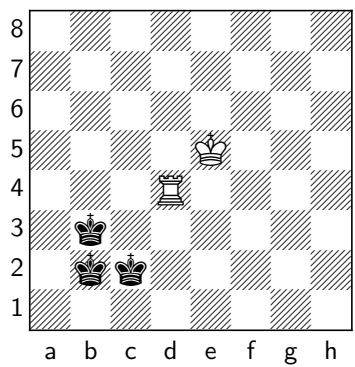

			\centering
			\newchessgame
			\setchessboard{
				showmover=false,
							boardfontsize=15pt,
				setpieces={Ke5, Rd4, kb3, kb2, kc2}
			}
			\chessboard
			\caption{Black moves to one of \pair{B}{3}, \pair{B}{2}, or \pair{C}{2}.}
		\end{figure}
		\begin{itemize}
			\item If the rook moves down or left, it is captured.
			\item Otherwise, the Black can return to \pair{C}{3} (safe space).
		\end{itemize}
	\end{enumerate}
	Hence, either the Black king cannot be forced to one edge or White cannot win.
\end{proof}

In three situations, the Black king was forced out of the safe space, and could not directly return.
Here, however, the Black king can always move to a cluster of squares, from which he can either return to the safe space, capture material of White (leading to a draw or victory for Black), or the cluster of squares reoccurs (thus Black survives in an endless loop).

We start with an important lemma, which we will use multiple times in the following.
\begin{lemma}[***]\label{lem:***}
	Let the Black king be in one of the following clusters of squares:
	\begin{enumerate}%[label=(\arabic*)]
		\item \pair{C}{1}, \pair{C}{2}, \pair{E}{1}, \pair{E}{2}
		\item \pair{D}{1}, \pair{D}{2}, \pair{F}{1}, \pair{F}{2}
		\item \pair{D}{1}, \pair{D}{2}, \pair{E}{1}, \pair{E}{2}
	\end{enumerate}
	Let the White king and the White rook be anywhere outside (1), (2), or (3), and they do not cover any of the clusters where the Black king is located.
	Assume that it is White's turn, then either the lemma reoccurs, Black captures a piece of White, or moves to the safe space.
\end{lemma}
\begin{proof}
	White is to move.
	\begin{itemize}
		\item If White moves the king, then the square combinations (1), (2), or (3) sustain themselves, or Black can take the White king. If the White king moves to the $B$ or $G$ file (which is safe when the rook covers $C$ resp. $F$) and blocks squares of some clusters; note that the rook cannot be protected by the king when moving to the left or right on line $3$ and will be catched, and that the Black king can advance to a safe space when the rook moves on other lines. Thus, the White king stays behind and covers the rook.
		\item If White moves the rook
		\begin{enumerate}
			\item left or right, it might uncover the Black kings location if placed between the $C$ to $F$ file. In any case, Black can in its next turn move the king to any square of another, undiscovered cluster. If the rook moves to the $C$ or $F$ file, we can assume that the Black king was in the invisible file. Then, or if the rook moves to the $D$ or $E$-file, Black can move the king. If, e.g, the rook is moved to $D$, then
				\begin{enumerate}
					\item if the Black king is in (2) or in (3), the Black king is on \pair{D}{2}. The Black king then moves into cluster (1), and this lemma reapplies.
					\item if the Black king is in (1), the Black king can stay there, and the assumptions remain unchanged.
				\end{enumerate}
			\item If the rook moves down to \pair{A-H}{1-2}, then move the Black king to the safe space, here: \pair{C-F}{3}, which is adjacent to (1), (2), and (3). This is possible because the White king can cover at most \pair{C-E}{3} or \pair{D-F}{3}. If the rook blocks the last one (\pair{C}{3} or \pair{F}{3}), the Black king can capture the rook or move to the safe space.
			\item If the rook moves up or down to \pair{A-H}{3-7}, the Black king can remain in his cluster.
		\end{enumerate}
	\end{itemize}
	Thus, either this lemma reoccurs, Black captures a piece, or moves to the safe space.
\end{proof}
 In a situation we have deferred in Lemma~\ref{kingpush}, White cannot guarantee a win.

\begin{lemma}[*]\label{lem:*}
	If the Black king is on \pair{A-B}{1-4} or \pair{C-D}{1-2}, it is White's turn, the White king is on \pair{D}{4}, and the rook is on \pair{D-H}{4-8}$\setminus \text{\pair{D}{4}}$.
	Then either this lemma reoccurs, Black captures a White piece, moves back to the safe space, or Lemma~\ref{lem:***} occurs.
\end{lemma}
\begin{proof}(Case distinction)
	\begin{figure}[!ht]
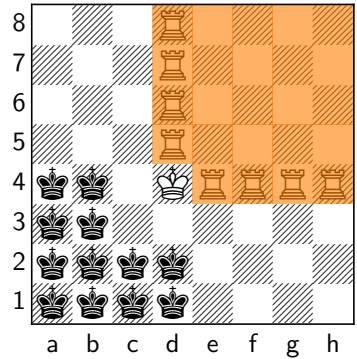

		\centering
		\newchessgame
		\setchessboard{
			showmover=false,
						boardfontsize=15pt,
			setpieces={Kd4, ka1,ka2,ka3,ka4, kb1,kb2,kb3,kb4, kc1,kc2, kd1,kd2, Rd5, Rd6, Rd7, Rd8, Re4, Rf4, Rg4, Rh4}
		}
		\chessboard[pgfstyle=color,opacity=0.6,color=orange, markarea=d5-h8, markarea=e4-h4]
		\caption{The orange squares indicate where the rook can be placed. The Black king can be on any of the squares \pair{A-B}{1-4}, or \pair{C-D}{1-2}. The White king is on \pair{D}{4}. The rook is on any of the orange squares. But there is only one rook on the board.}
	\end{figure}
	\begin{enumerate}
		\item Moving the king leads to the capture of it, or Black moves to the safe space (\pair{C}{3}).
		\item Moving the rook to the right or up allows us to reapply this Lemma.
		\item Move the rook left to the $A-C$-file. Black then might not have been on the now visible squares, but might be on \pair{D}{1-2}. It is now Black's turn, and the uncertainty extends to \pair{E}{1}, \pair{E}{2}. Thus, Lemma~\ref{lem:***} applies.
		\item Moving the rook down can be treated analogously to (3) or by symmetry.\qedhere
	\end{enumerate}
\end{proof}
The following deferred case-distinction concludes the overall strategy.
\begin{lemma}[**]\label{lem:**}
 Let it be White's turn with the White king be on any square of $\text{\pair{A-H}{1-8}} \setminus \text{\pair{A-E}{1-3}}$ (not adjacent to the Black king) and the rook on any on the squares of \pair{E-H}{3-8}.
 Let the Black king be anywhere on \pair{A-D}{1-2}. Then either this lemma reoccurs, Black captures a piece of White, or leads to Lemma~\ref{lem:***}.
\end{lemma}
\begin{proof} (Case distinction)

	\begin{enumerate}
		\item Let the White king be on any square in ranks 5--8 \pair{A-H}{5-8}.
		\begin{enumerate}
			\item If the rook is on \pair{E-H}{3}:
			\begin{itemize}
				\item If the king is moved, the Black king can remain everywhere on \pair{A-D}{1-2} and this lemma reapplies.
				\item Moving the rook
				\begin{itemize}
					\item up or down allows Black to return to the safe space.
					\item left or right; then either the rook can be captured, or the Black king can stay in \pair{A-D}{1-2}, and this lemma reapplies.
				\end{itemize}
			\end{itemize}
			\item If the rook is on \pair{E-H}{4-8}:
			\begin{itemize}
				\item If the king moves, this lemma reapplies as above.
				\item Moving the rook
				\begin{itemize}
					\item up lets Black move to the safe space \pair{C}{3}.
					\item down to \pair{E-H}{3-7} leads to the Black king can stay on \pair{A-D}{1-2} and this Lemma reapplies.
					\item down to \pair{E-H}{1-2} let Black move to the safe space.
					\item to the left or right allows the Black king to move to the safe space.
				\end{itemize}
			\end{itemize}
		\end{enumerate}
		\item Let the White king be on \pair{A-C}{4}.
		\begin{enumerate}
			\item If the rook is on \pair{E-H}{4-8}:
			\begin{itemize}
				\item Moving the king leads to either the king being captured, or the Black king can stay on \pair{A-D}{1-2} and the lemma reapplies.
				\item Moving the rook
				\begin{itemize}
					\item down either results in capture of the rook or Black can move to the safe space \pair{E}{3}.
					\item up or right let Black move to \pair{E}{3} (safe space) or reapply the lemma.
					\item left either let Black move to the safe space (\pair{E}{3}) or the Black king can stay in \pair{A-D}{1-2} and the lemma reapplies.
				\end{itemize}
			\end{itemize}
			\item If the rook is on \pair{E-H}{3}:
			\begin{itemize}
				\item Moving the king either allows Black to capture it, or the Black king can stay on \pair{A-D}{1-2} and this lemma can reapply.
				\item Moving the rook
				\begin{itemize}
					\item down either results in capture or allows Black to move to the safe space \pair{E}{3}.
					\item  up or right allows the Black king to stay on \pair{A-D}{1-2} and reapply this lemma.
					\item  to the left:
					\begin{itemize}
						\item Moving it to \pair{E-G}{3} allows Black to stay on \pair{A-D}{1-2} and reapply this lemma.
						\item Moving it to \pair{A-C}{3}: Black chooses not to be there, extending the uncertainty to \pair{E}{1}, \pair{E}{2}, while maintaining its position on \pair{D}{1}, \pair{D}{2}, leading to \hyperref[lem:***]{Lemma~(***)}.
						\item Moving it to \pair{D}{3}: Black chooses to be on \pair{D}{2}, then moves to \pair{C}{1}, \pair{C}{2}, \pair{E}{1}, \pair{E}{2}, leading again to \hyperref[lem:***]{Lemma~(***)}.
					\end{itemize}
				\end{itemize}
			\end{itemize}
		\end{enumerate}
		\item Let the White king be on \pair{E-H}{4}. The cases are similar to 2, except that Black moves to the safe space square \pair{C}{3} instead of e.g. \pair{E}{3}.
		\item Let the White king be on \pair{D}{4} (\textbf{dangerous (!)} because the White king blocks the reachable squares of the safe space).
		\begin{itemize}
			\item Moving the king either allows Black to capture it, or the Black king can stay on \pair{A-D}{1-2}, and this lemma can reapply.
			\item Moving the rook left, up, or right results in \hyperref[lem:***]{Lemma~(***)}, the rook being captured, or the Black king can stay on \pair{A-D}{1-2} and this lemma reapplies.
			\item Moving the rook down to \pair{E-H}{3-7}, then the Black king can stay on \pair{A-D}{1-2} and this lemma reapplies.
			\item Moving the rook down to \pair{E}{1-2} allows Black to take the rook.
			\item Moving the rook down to \pair{F-H}{1-2}:
			\begin{itemize}
				\item to \pair{F}{2}: Black chooses to be on \pair{B}{2}, and can then move to \pair{A-C}{1}, or \pair{A-B}{3}.
				\begin{itemize}
					\item Move the White king
					\begin{itemize}
						\item Left/ Left-down, then capture.
						\item down: this situation is described as $(X1)$ at the end of this proof.
						\item anywhere else: Move to safe space (\pair{C}{3}).
					\end{itemize}
					\item Move the rook
					\begin{itemize}
						\item up: Then Black decides to be on \pair{A-C}{1}, then expanding to \pair{A-D}{2}, so this lemma reapplies.
						\item left/ right: either Black can take the rook, or Black decides to be on \pair{A}{3}, \pair{B}{3} and expands to \pair{A}{4}, \pair{B}{4} (Situation~$(X2)$, see below).
						\item down, Black decides to be on \pair{A-B}{3} and expands to \pair{A-B}{2-4} (Situation~$(X3)$, see below).
					\end{itemize}
				\end{itemize}
				\item The squares \pair{G}{2}, \pair{H}{2} behave analogously.
				\item Move the rook to \pair{F}{1}. Then Black chooses not to be on \pair{A-D}{1}. So, black can now be on \pair{A-D}{2}, \pair{A}{3}, \pair{B}{3}, and \emph{E2} (!).
				\begin{itemize}
					\item If White now moves the rook up (more than one square), then either Black takes or we can apply \hyperref[lem:***]{Lemma~(***)} as the king can now expand to \pair{D}{1} and \pair{E}{1}.
					\item If White moves the rook up one square, then the previous case with the rook on \pair{F-H}{2} can be applied as Black might be on \pair{B}{2}. So, the Black king can stay on \pair{A-D}{1-2} and this lemma reapplies.
					\item If White moves anything anywhere else, then Black can move to the safe space, as he can be on \pair{E}{2} or capture.
				\end{itemize}
				\item Analogously for \pair{G}{1}, \pair{H}{1}.
			\end{itemize}
		\end{itemize}
		\item If the king is on \pair{F-H}{1-3}
		\begin{itemize}
			\item Move the king: capture or Black can stay on \pair{A-D}{1-2} and reapply this lemma.
			\item Move the rook: 
			\begin{itemize}
				\item down to \pair{E-H}{1-2}: move to safe space.
				\item left to \pair{A-D}{3}: take the rook.
				\item left to \pair{A-D}{4-8}: move to safe space.
				\item otherwise the Black king can stay on \pair{A-D}{1-2} and reapply this lemma.
			\end{itemize}
		\end{itemize}
	\end{enumerate}
	
			\begin{figure}[]
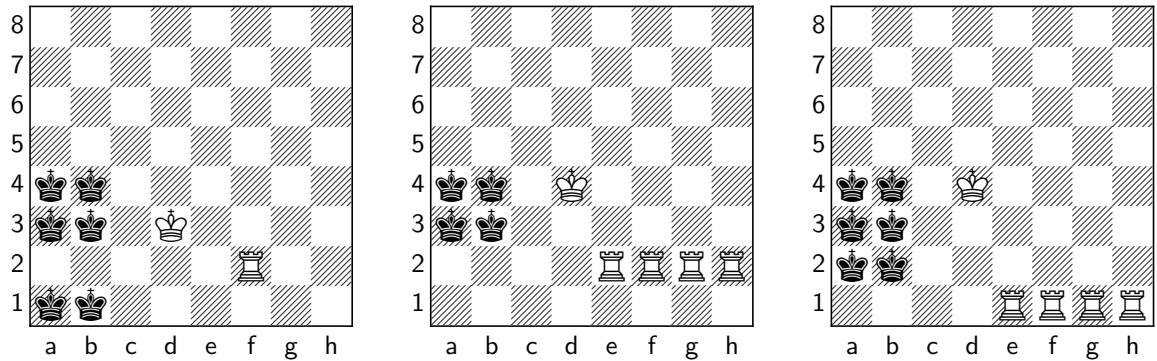

		\centering
		\newchessgame
		\setchessboard{
	showmover=false,
	boardfontsize=15pt,
	setpieces={Rf2, ka1,kb1,kc1kd1,ka3,kb3,ka4,kb4, Kd3}
}
\chessboard
		\setchessboard{
	showmover=false,
	boardfontsize=15pt,
	setpieces={Kd4, Re2, Rf2, Rg2, Rh2, ka3,ka4,kb3,kb4}
}
\chessboard
		\setchessboard{
	showmover=false,
	boardfontsize=15pt,
	setpieces={Kd4, Re1,Rf1,Rg1,Rh1,ka2,ka3,ka4,kb2,kb3,kb4}
}
\chessboard
		\caption{The situation described as (X1), (X2), and (X3) (from left to right)}
	\end{figure}

	(X1): Then Black can move to any of \pair{A-D}{1} and \pair{A-B}{3-4}. Now, White is to move.
	\begin{itemize}
		\item If the White king moves, then either capture it or \hyperref[lem:***]{Lemma~(***)} (\pair{A}{4}, \pair{A}{5}, \pair{B}{4}, \pair{B}{5} which by rotating the board results in \hyperref[lem:***]{Lemma~(***)}).
		\item If the rook moves:
		\begin{itemize}
			\item down; then \hyperref[lem:***]{Lemma~(***)} by moving to \pair{A-B}{5}.
			\item left; then take or \hyperref[lem:***]{Lemma~(***)}.
			\item right; then ends up in \hyperref[lem:***]{Lemma~(***)}.
			\item up to \pair{F}{3-4} (one or two squares): then Black decides to be on \pair{B}{4}, and expands to \pair{A-B}{3} and \pair{A-B}{5}, resulting in \hyperref[lem:***]{Lemma~(***)}. Otherwise, results in this \hyperref[lem:**]{Lemma~(**)} by moving to \pair{A-B}{1-4}, rotate board.
		\end{itemize}
	\end{itemize}
	
	(X2): It is White's turn.
	
	\begin{itemize}
		\item Moving the king, either take, or \hyperref[lem:***]{Lemma~(***)}.
		\item Move the rook:
		\begin{itemize}
			\item down, left, right, then take or \hyperref[lem:***]{Lemma~(***)}.
			\item up three squares to \pair{E-H}{5} (otherwise \hyperref[lem:***]{Lemma~(***)}): expand to \pair{A-D}{1-2} which leads to \hyperref[lem:***]{Lemma~(***)}.
		\end{itemize}
	\end{itemize}
	
	(X3): It is White's turn.
	
	\begin{itemize}
		\item Move king: either Black captures or \hyperref[lem:***]{Lemma~(***)} (could also move to safe-square).
		\item Move rook: Up four squares, then this lemma reoccurs by expanding to \pair{A}{1}, \pair{B}{1}, otherwise: capture or \hyperref[lem:***]{Lemma~(***)}\qedhere
	\end{itemize}
\end{proof}

\begin{corollary}
	If White starts with a king and a rook, and Black only has a king. The Black king is not in check and anywhere on \pair{C-F}{3-6}, and it is White's turn. Then White cannot force a win.
\end{corollary}

\section{Two Rooks and King vs. King}
Finally, we consider the endgame of a king and two rooks against a king.
As before, let White have a king and two rooks, and Black only a king. White starts and knows the starting position of Black.
We prove that for any starting position of the Black and White figures, there exists a winning strategy for White.
Also in classical chess, this is a win for White \cite{rook_win}.

The strategy consists of three stages, organized in three lemmas below. In the first stage, we will prove that White can move the two rooks on one rank/file while ensuring that the own king is not on it as well. 
Secondly, we will show that we can move the White king into one corner and the two rooks are adjacent to him. 
And lastly, we will prove that the two rooks and the king can force a checkmate from the edge position.

Let us start with the first stage: we move the two rooks such that they cover each other, and their own king is not on either the same rank or file.
\begin{lemma}\label{twoRooks1}
	Given that White starts with a king and two rooks, and Black only a king. It is White's turn. White knows the starting position of the Black king.
	
	Then within at most two moves, White can move the rooks such that they cover each other without their own king being on either the same rank or file as both of the rooks.
\end{lemma}
\begin{proof}
	 Note that when Black starts on $(X,Y)$, there is at most a $5 \times 5$-square \pair{[X-2]-[X+2]}{[Y-2]-[Y+2]} in which Whites figures are at risk in Blacks second move.

	Therefore, there are at least two ranks/files available which are safe from the opponent's king after its first move and without the own White king. Note that the Black king cannot block any move of the rooks, as if he could, the rooks would be able to capture him. 
	
	In the first two moves, White can move the two rooks to one of the ranks/files. 
\end{proof}
Now we will see that the two rooks can assist their own king to move to a corner.
\begin{lemma}
	Given a situation after Lemma~\ref{twoRooks1}. Then White can move its own king in a corner, and place the two rooks are adjacent to him.
\end{lemma}
\begin{proof}
	Note that rooks can move freely on the rank/file where the White king is not. This can be done, as the rooks cover each other, and their own king cannot block them.
	
	White moves their own king towards an edge by repeating the following strategy:
	\begin{itemize}
		\item Cover two files/ranks ahead of the own king, such that he can move freely there. If there is only one file/rank ahead, cover that one.
		\item Move the own king one square towards one edge (not diagonal). 
	\end{itemize}
	
	The second step can be done, as the rooks cover all squares which are adjacent to the new position of the own king (which are not covered by the king itself). One might note that diagonal moves are not always safe, and that it is possible for the Black king to \emph{hide} behind the White king such that he does not get seen by the rooks.

	Thus, when the White king arrives at the edge, he cannot move away from the rooks without the risk of running into Black's king. Instead, he has to move towards his own rooks.
	
	When the White king arrives at the field adjacent to the corner (which is blocked by a rook), it is easy to see that White can safely rotate its own pieces to free the corner from the rook and move the king here and the rooks adjacent to him.
\end{proof}
Finally, we see that the king and two rooks can finish by forcing a checkmate. 
\begin{lemma}
	Assume White starts with a king and two rooks, and Black only a king. It is White's turn.
	Suppose that the king is in one corner, and the two rooks are adjacent to him. 
	
	Then there is a strategy such that White always wins.
\end{lemma}
\begin{proof}
	Without loss of generality, let the White king be on \pair{A}{1}, and the two rooks on \pair{A}{2} (Rook 1) and \pair{B}{1} (Rook 2).
		The following strategy resembles the staircase mate in classical chess, but we have to be more careful, as we do not know the exact position of the Black king. Therefore, we have to cover the rooks with the White king such that the Black king cannot take one of them. The strategy proceeds as follows and can be repeated until checkmate:
	
		\begin{itemize}
		\item Move Rook 1 up one square.
		\item Move Rook 2 up one square.
		\item Move the White king one square up.
	\end{itemize}
	As in classical chess, this forces the Black king to move up (and right), as otherwise White would capture him. Eventually, the Black king has no moves left to escape.
\end{proof}
\begin{corollary}
	Given that White starts with a king and two rooks, and Black with a king. White begins and knows the initial position of the Black king.	
	Then White wins.
\end{corollary}
\section{Final Comments and Open Questions}
We showed that in Fog of War chess, given that one player has a king and queen and plays against a lone king, he can always guarantee a win. In contrast to classical chess, if one player has a king and a rook, then he cannot guarantee a win against a player with a lone king. However, adding one more rook for the player with the rook guarantees a win for him. 

In Fog of War chess, there exist many more open questions concerning endgames:

First and foremost, endgames using knights, bishops, or a combination of several pieces are yet to be answered. 
As in classical chess, endgames can also occur where the opponent has not just the king, but also a few more own pieces available.

Another relevant question is to minimize the number of moves one player needs to make to ensure a win. We suspect that our strategy for king and queen versus king cannot be performed within 50 moves, which might be relevant when one plays with a fifty-move rule. Potentially, another strategy could guarantee a win faster and within 50 moves. Our strategy for king and two rooks versus king however, should work within 50 moves.

Finally, it is worth mentioning that also other variants with limited knowledge about the board than the classical Fog of War variant might yield interesting results. 

In conclusion, Fog of War chess presents a rich field for theoretical exploration, with many open questions and potential for further research.

We hope that our work contributes to the understanding of endgames in Fog of War chess and opens up avenues for future research in this intriguing variant of chess.

\bibliography{bibliography.bib}

\appendix
\end{document}